\documentclass[aps,pra,twocolumn,showpacs,superscriptaddress, nofootinbib]{revtex4-2}

\pdfoutput=1

\usepackage[utf8]{inputenc}
\usepackage{mathpazo}
\usepackage[T1]{fontenc}
\usepackage{bbm}
\usepackage{soul}
\usepackage[normalem]{ulem}
\usepackage{braket}
\usepackage{graphicx}
\usepackage{amsmath}
\usepackage{xcolor}

\usepackage[colorlinks,citecolor=blue,linkcolor=blue]{hyperref}

\usepackage{qcircuit}

\makeatletter
\newsavebox{\@brx}
\newcommand{\llangle}[1][]{\savebox{\@brx}{\(\m@th{#1\langle}\)}%
  \mathopen{\copy\@brx\kern-0.5\wd\@brx\usebox{\@brx}}}
\newcommand{\rrangle}[1][]{\savebox{\@brx}{\(\m@th{#1\rangle}\)}%
  \mathclose{\copy\@brx\kern-0.5\wd\@brx\usebox{\@brx}}}
\makeatother

\newtheorem{theorem}{Theorem}[section]

\newenvironment{proof}[1][Proof]{\begin{trivlist}
\item[\hskip \labelsep {\bfseries #1}]}{\end{trivlist}}

\DeclareMathOperator{\tr}{tr}
\DeclareMathOperator{\id}{\mathbb{I}}

\newcommand{\HH}{\mathcal{H}}

\newcommand{\ketbra}[2]{| {#1} \vphantom{#2} \rangle\langle {#2} \vphantom{#1} |} 
\newcommand{\proj}[1]{\ketbra{#1}{#1}}

\newcommand{\qed}{\nobreak \ifvmode \relax \else
      \ifdim\lastskip<1.5em \hskip-\lastskip
      \hskip1.5em plus0em minus0.5em \fi \nobreak
      \vrule height0.75em width0.5em depth0.25em\fi}

\begin{document}
\title{A no-go theorem for the persistent reality of Wigner's friend's perception}
\author{Philippe Allard Gu\'{e}rin}
\email{pallardguerin@perimeterinstitute.ca}
\affiliation{Faculty of Physics, University of Vienna, Boltzmanngasse 5, 1090 Vienna, Austria}
\affiliation{Institute for Quantum Optics and Quantum Information (IQOQI), Austrian Academy of Sciences, Boltzmanngasse 3, 1090 Vienna, Austria}
\affiliation{Perimeter Institute for Theoretical Physics, 31 Caroline St. N, Waterloo, Ontario, N2L 2Y5, Canada}

\author{Veronika Baumann}
\affiliation{Faculty of Physics, University of Vienna, Boltzmanngasse 5, 1090 Vienna, Austria}
\affiliation{Institute for Quantum Optics and Quantum Information (IQOQI), Austrian Academy of Sciences, Boltzmanngasse 3, 1090 Vienna, Austria}

\author{Flavio Del Santo}
\affiliation{Faculty of Physics, University of Vienna, Boltzmanngasse 5, 1090 Vienna, Austria}
\affiliation{Institute for Quantum Optics and Quantum Information (IQOQI), Austrian Academy of Sciences, Boltzmanngasse 3, 1090 Vienna, Austria}

\author{\v{C}aslav Brukner}
\affiliation{Faculty of Physics, University of Vienna, Boltzmanngasse 5, 1090 Vienna, Austria}
\affiliation{Institute for Quantum Optics and Quantum Information (IQOQI), Austrian Academy of Sciences, Boltzmanngasse 3, 1090 Vienna, Austria}

\date{\today}

\begin{abstract}

The notorious Wigner's friend thought experiment (and modifications thereof) has in recent years received renewed interest especially due to new arguments that force us to question some of the fundamental assumptions of quantum theory. In this paper, we formulate a no-go theorem for the persistent reality of Wigner's friend's perception, which allows us to conclude that the perceptions that the friend has of her own measurement outcomes at different times cannot "share the same reality'', if seemingly natural quantum mechanical assumptions are met. More formally, this means that, in a Wigner's friend scenario, there is no joint probability distribution for the friend's perceived measurement outcomes at two different times, that depends linearly on the initial state of the measured system and whose marginals reproduce the predictions of unitary quantum theory. This theorem entails that one must either (1) propose a nonlinear modification of the Born rule for two-time predictions, (2) sometimes prohibit the use of present information to predict the future --thereby reducing the predictive power of quantum theory-- or (3) deny that unitary quantum mechanics makes valid single-time predictions for all observers. We briefly discuss which of the theorem's assumptions are more likely to be dropped within various popular interpretations of quantum mechanics.

\end{abstract}

\maketitle



\section{Introduction}
\label{sec:intro}

One of the most puzzling scenarios one can encounter in quantum physics is the so-called Wigner's friend thought experiment~\cite{Wigner1995, Deutsch1985}. 
It allows to investigate the applicability of the quantum formalism beyond its usual limits, by considering a certain physical system --called ``the friend''-- simultaneously as a quantum system and as a user of quantum theory (observer). In this thought experiment, a superobserver (Wigner) describes, using the pure quantum formalism, his friend who is performing a quantum measurement on a spin system. After the friend's measurement has taken place, we are in a counter-intuitive situation where Wigner describes the friend in a quantum superposition of observing two different outcomes, while from the friend's perspective a definite outcome must be perceived.

There has been a number of recent works that cast new light onto this thought experiment \cite{Brukner2017, Brukner2018, relano2018decoherence, Baumann2018, healey2018quantum, nurgalieva2018inadequacy, Bong2020, sudbery2019hidden, Baumann2019, Baumann2019generalized, baumann2020wigner, bub2020two, Cavalcanti2020, Proietti2019, Baumann2019, Zukowski2020}, many of which originated as reactions to a paper by Frauchiger and Renner~\cite{Frauchiger2018}. The latter work can be regarded as showing that, in quantum mechanics, it may be problematic to treat observational knowledge of other agents as if it were one's own, and to logically compare such indirect knowledge with that gained through direct observation. In the words of these authors, in a scenario where ``multiple agents have access to different pieces of information, and draw conclusions by reasoning about the information held by others'', it can be shown that, ``in the general context of quantum theory, the rules for such nested reasoning may be ambiguous''~\cite{Frauchiger2018}; a conclusion that is reminescent of the QBist interpretation of quantum theory~\cite{Fuchs2010, DeBrota2020}. Other works~\cite{Brukner2017, Brukner2018, Bong2020, Cavalcanti2020}, show that a no-go result can be obtained when the assumption that superobservers can treat other observer's outcomes as "facts of the worlds" is combined with a locality assumption. What all of the above-cited works have in common is that they reach their no-go results by combining the observations of multiple observers as if those all belonged to the same "classical reality".

 In the present paper, we put forward a \textit{no-go theorem for the persistent reality of Wigner's friend's perception} that has perhaps more counter-intuitive and drastic consequences: Even a single observer, when making predictions about his or her observations at two different times, can conflict with the linear dependence of quantum mechanical probabilities on the density operator of the system being measured. This will occur if said observer is subjected to a measurement by a superobserver between these two times and uses unitary quantum mechanics (i.e. "no-collapse" quantum mechanics) to make his or her predictions. Our result is in line with a different understanding of the Frauchiger-Renner argument, where it is taken to primarily show that inferences made on the basis of a quantum state assigned at a certain time are not necessarily valid at later times, especially not if ``someone Hadamards your brain'' in between~\cite{Aaronson_blog}.

Indeed, one conclusion that can be drawn from this no-go theorem is that treating a piece of information from the past as if it was still presently existing (even when one takes into account a possible subjective uncertainty) cannot in general be upheld together with the conjunction of the above seemingly natural assumptions within the domain of quantum theory. We will show that in a particular instance of the Wigner's friend experiment, our assumptions imply that, even in cases where the theory states that no change would take place in the quantum state of the friend's laboratory, the perceptions of the friend have a non-zero probability to change from before to after Wigner's measurement. Finally, we will briefly discuss how different interpretations of quantum mechanics might comply with the no-go theorem, by identifying which of the assumptions are most likely to be abandoned within these interpretations.


 \section{Wigner's friend thought experiment}
\label{sec:gedanken}

Let us begin by reviewing the Wigner's friend thought experiment within a unitary formalism to set up some basic notation, and to clarify what we mean by \textit{unitary quantum mechanics} for the particular scenario that we consider. It is common to assume that the following description provided by unitary quantum mechanics is empirically adequate in all situations: any observer which uses this formalism will predict probabilities that match the relative frequencies that would be observed if the experiment was repeated many times.

The experiment features an observer, the friend (F), performing measurements on a qubit (e.g., a spin-1/2 particle), the system (S), in a sealed laboratory. The system is initialized in the state $|\psi \rangle_S = \alpha \ket{\uparrow}_S + \beta \ket \downarrow_S$, where $\alpha$ and $\beta$ are complex numbers that obey $|\alpha|^2 + |\beta|^2 = 1$, and the possible outcomes of the measurement are recorded by the friend as  $U$ or $D$, respectively standing for "up" and "down". 
Another (super-)observer, Wigner (W), located outside the laboratory, performs a measurement on both the system and the friend. The initial state of the friend (which encompasses any other possible degree of freedom in the isolated lab) is known by Wigner, and is initially in a macroscopic "ready" state $|0 \rangle_F$. 
The state of Wigner himself is also in a macroscopic ``ready'' state $|0 \rangle_W$.

The initial state (at time $t_0$) of the whole setup is therefore
\begin{equation}
\label{eq:psi_t0}
|\Psi (t_0) \rangle = |\psi \rangle_S  |0 \rangle_F |0 \rangle_W = \left(\alpha \ket \uparrow_S + \beta \ket \downarrow_S \right) |0 \rangle_F |0 \rangle_W.
\end{equation}
At time $t_1$, the friend measures the spin in the z-basis, and the state becomes
\begin{align}
\label{eq:psi_t1}
|\Psi (t_1) \rangle &=\left( \alpha \ket \uparrow_S |U \rangle_F  +  \beta \ket \downarrow_S |D \rangle_F\right) |0 \rangle_W ,
\end{align}
where the states $\ket{U}_F, \ket{D}_F$ correspond to the friend having observed outcome ``up'' or ``down'' respectively. Later, at time $t_W$, Wigner measures the friend and system in some entangled basis, with binary outcomes~\footnote{Strictly speaking, there should be two other outcomes corresponding to the rank-2 projector $|\uparrow, D\rangle \langle \uparrow, D| + |\downarrow, U\rangle \langle \downarrow, U|$, but these outcomes are never actualized in the experiment.} corresponding to the orthogonal states
\begin{align}
&|1\rangle_{SF} = a \ket{\uparrow, U} + b \ket{\downarrow, D}\nonumber \\
&|2 \rangle_{SF} = b^* \ket{\uparrow, U}- a^*  \ket{\downarrow , D}\nonumber,
\end{align}
with $a, b$ being complex numbers obeying $|a|^2 +|b|^2 = 1$. At a slightly later time $t_2 > t_W$, the measurement is over and we have the final state
\begin{align}
|\Psi(t_2)\rangle = &(\alpha a^* + \beta b^*) |1\rangle_{SF} |1\rangle_W +  (\alpha  b - \beta a) |2\rangle_{SF} | 2 \rangle_W \label{eq:psi_t2}  \nonumber \\
=& a (\alpha a^* + \beta b^*) \ket \uparrow_S |U \rangle_F |1\rangle_W  \nonumber \\
&+  b (\alpha a^* + \beta b^*) \ket \downarrow_S |D \rangle_F |1\rangle_W  \\
&+ b^* (\alpha  b - \beta a) \ket \uparrow_S |U \rangle_F | 2 \rangle_W  \nonumber \\
& -  a^* (\alpha  b - \beta a) \ket \downarrow_S |D \rangle_F | 2 \rangle_W \nonumber,
\end{align}
where $|1\rangle_W$ and $|2\rangle_W$ are pure quantum states corresponding to Wigner seeing the outcome ``$1$'' or ``$2$'' respectively. Note that the state $|\Psi(t_2)\rangle$ depends on the specific unitary realization of Wigner's measurement; different purifications can lead to different states $|\Psi(t_2)\rangle$.

Using the states in Eqs.~\eqref{eq:psi_t0}-\eqref{eq:psi_t2} and the Born rule, one can find the expected statistics for any of the friend's or Wigner's measurement outcomes using unitary quantum mechanics. This is achieved  by applying a projector $\Pi_x$ onto the state where the respective observer is seeing outcome $x$ to the state at the time of interest, i.e.
\begin{equation}
p(x) = \tr\left( \Pi_x \proj{\Psi(t)} \right),
\label{eq:p_onetime}
\end{equation}
where, in the case relevant for this work where there are two outcomes $U$ and $D$, the probability of the friend seeing for example outcome $x = U$ is obtained with $\Pi_U = |U \rangle \langle U|_F$.

The states in Eqs.~\eqref{eq:psi_t0}-\eqref{eq:psi_t2} represent the unitary evolution of the full quantum state at all times. While the latter is commonly associated with the many-worlds interpretation~\cite{Everett1957,wallace2012emergent}, or with Bohmian mechanics~\cite{Bohm1952_1, Bohm1952_2, durr2009bohmian}, it is also compatible with a timeless formulation of quantum theory as introduced by Page and Wootters~\cite{pageEvolutionEvolutionDynamics1983}. Even without necessarily accepting the picture of the world provided by the many-worlds interpretation, Eq.~\eqref{eq:p_onetime} can be used by any observer to make predictions about their observered outcome at some time. In particular, we assume that the friend has enough information about the experimental setup in order to use Eq.~\eqref{eq:p_onetime} for her probability assignments.

In the following, we will also be interested in cases where the initial state of the system is a mixed state $\rho_S$. Such a state can be decomposed as $\rho = \lambda \proj{\psi} + (1 - \lambda) \proj{\phi}$, where $|\psi\rangle, |\phi\rangle$ are orthonormal states and $0 \leq \lambda \leq 1$. Then we have the analogue of expressions~\eqref{eq:psi_t0}-\eqref{eq:psi_t2} for the mixed state $\Sigma(t)$ of the whole setup at different times
\begin{equation}
\Sigma(t) = \lambda \proj{\Psi(t)} + (1 - \lambda) \proj{\Phi(t)},
\label{eq:Sigma_t}
\end{equation}
where $|\Psi(t)\rangle$ and $\ket{\Phi(t)}$ are states analogous to Eqs.~\eqref{eq:psi_t0}-\eqref{eq:psi_t2} with initial system states $\ket{\psi}_S$ and $\ket{\phi}_S$ respectively. Furthermore, probabilities are now given by $p(x) = \tr\left( \Pi_x \Sigma(t) \right)$.

In the standard analysis of the Wigner's friend thought experiment, the friend is usually assumed to describe the dynamics of her lab by using the state-update rule instead of Eq.~\eqref{eq:psi_t1}. She, therefore, would assign probabilities to Wigner's measurement that are different from those assigned by Wigner based on Eqs.~\eqref{eq:psi_t1} and~\eqref{eq:psi_t2}, which leads to an inconsistency between the predictions of both observers (see for example Ref.~\cite{Baumann2018}).

\section{Probability assignments in a scientific theory}
\label{sec:probs}

A necessary requirement for an empirically adequate scientific theory is that it should be able to give (quantitative) predictions, such as to have testable empirical content.\footnote{Throughout the paper, by the term \textit{prediction} we merely mean the possibility of assigning a probability distribution, and we do not strictly commit to any particular (operational) interpretation of probability, such as in terms of betting quotients.}  Namely, a theory should be able to associate a measure of likelihood to an event $y$ to happen, given that certain conditions --that in turn are captured by another event $x$-- have already occurred. In the words of Wigner,
\begin{quote}
"One realises that \textit{all} the information which the laws of physics provide consists of probability connections between subsequent impressions that a system makes on one if one interacts with it repeatedly, i.e., if one makes repeated measurements on it."~\cite{Wigner1995}
\end{quote}

The theory should thus be able to answer questions of the form: ``given that I observed event $x$ at time $t_1$, how likely is it that I will observe event $y$ at a later time $t_2$?''. Mathematically speaking, such a question is answered by specifying a conditional probability distributions $p(y|x)$.~\footnote{In situations where a user of the theory does not have enough information to uniquely determine $p(y|x)$ (for example some other agent could intervene in an unknown way), the theory should be able to provide a list of all the information which, if it were known, would determine $p(y|x)$. For lack of a better alternative, the user of the theory can subjectively assign their best guess for a probability distribution over these unknown variables, which in turn allows to compute $p(y|x)$.} On that note, a recent work about the emergence of physical laws is based on the idea that the primary purpose of such laws is to give the conditional probability distributions relating events perceived by an observer at two subsequent times~\cite{Mueller2020}.

In the context of the Wigner's friend thought experiment, we are thus interested in the friend's question: "given that I saw outcome $f_1$ at time $t_1$, what is the probability (attributed by using quantum theory) that I will see outcome $f_2$ at time $t_2$?". We assume that quantum mechanics is empirically adequate and is able to answer this question by providing a conditional probability distribution $p(f_2|f_1)$. Moreover, note that (unitary) quantum theory also prescribes how to assign a probability for observations at a single time (\emph{one-time probabilities}) $p(f_1)$ by  Eq.~\eqref{eq:p_onetime}.~\footnote{It is worth stressing that one-time probabilities are fundamentally also themselves conditional probabilities, namely conditioned on all the possible past events that can influence the probability of the event that we are trying to predict.} Given these elements, the standard axiomatization of probability theory allows the definition of a joint probability distribution through the identity $p(f_1, f_2)=p(f_1|f_2)p(f_2)$. Thus, any theory that, like quantum mechanics, prescribes rules to assign one-time probabilities and conditional probabilities, automatically allows to assign joint probability distributions. However, as we will see, such a joint probability distribution cannot simultaneously fulfill three seemingly natural assumptions in a Wigner's friend scenario.

\section{No-go theorem for the persistent reality of Wigner's friend's perception}

We now formulate a formal no-go theorem which shows that in Wigner's friend scenarios, the friend cannot treat her perceived measurement outcome as having \textit{reality across multiple times} without contradicting what might appear to be core assumptions of quantum mechanics. Consider the following assumptions:

\begin{enumerate}
\item [P1] The events $f_1$ and $f_2$, corresponding to the perceived measurement records of the friend at times $t_1$ and $t_2$, respectively, can be combined into a joint event to which is assigned a probability distribution $p(f_1, f_2)$. Moreover, the rules of the probability calculus imply that $p(f_1) = \sum_{f_2} p(f_1, f_2)$ and $p(f_2)= \sum_{f_1} p(f_1, f_2)$.
\item[P2] One-time probabilities are assigned without resorting to the state-update rule (i.e., using unitary quantum theory, where no ``collapse'' is considered to occur). Thus, when the initial state of the qubit is $\ket{\psi}_S$,
\begin{equation}
p(f_i) = \tr ( |f_i \rangle \langle f_i |_F |\Psi(t_i) \rangle \langle \Psi(t_i)|),
\label{eq:p_fi}
\end{equation}
with $|\Psi(t_i) \rangle$ being the unitarily evolved global state according to Eqs.~\eqref{eq:psi_t1},~\eqref{eq:psi_t2}.
\item[P3] The joint probability of the friend's perceived outcomes $p(f_1, f_2)$ has a convex linear dependence on the initial state $\rho_S$ of the system qubit.
\end{enumerate}

We will motivate these assumptions in more detail in Section~\ref{sec:assumptions}. We now show that these assumptions lead to a contradiction when applied to the friend in a Wigner's friend scenario.
\begin{theorem}
\label{thm:no_joint}
The conjunction of the assumptions P1-P3 cannot be satisfied for the Wigner's friend thought experiment for a general choice of Wigner's measurement basis.
\end{theorem}

\begin{proof}
Define the isometries $V_i : \HH_S \to \HH_S \otimes \HH_F \otimes \HH_W$, $i = 1,2$ mapping the initial state of the spin $|\psi\rangle_S$ to the corresponding state at time $t_i$ as $V_i |\psi\rangle_S = |\Psi(t_i)\rangle_{SFW}$. Using Eqs.~\eqref{eq:psi_t1} and~\eqref{eq:psi_t2}, these are found to be
\begin{align}
V_1 &= |\uparrow, U, 0 \rangle_{SFW} \bra \uparrow_S + |\downarrow, D, 0 \rangle_{SFW} \bra \downarrow_S \\
V_2 &= |1 \rangle_{SF}|1\rangle_W \langle \phi_1|_S + |2 \rangle_{SF}|2\rangle_W \langle \phi_2|_S
\end{align}
where $|\phi_1\rangle := a \ket \uparrow+ b \ket \downarrow $ and $|\phi_2\rangle := b^* \ket \uparrow - a^* \ket \downarrow $.

By P2 (and using P3 to extend to mixed states) we have
\begin{align}
p(f_1)& = \tr \left( (|f_1 \rangle \langle f_1|_F \otimes \id_{SW}) V_1 \rho V_1^\dagger \right) \\
&= \tr \left( V_1^\dagger (|f_1 \rangle \langle f_1|_F \otimes \id_{SW}) V_1 \rho \right) = \tr( E^1_{f_1} \rho) \label{eq:p_f1_V},
\end{align}
where we define $E^1_{f_1} :=  V_1^\dagger (|f_1 \rangle \langle f_1|_F \otimes \id_{SW}) V_1$, which can be understood as the "Heisenberg picture" operator~\footnote{This is not exactly the textbook Heisenberg picture, because $V_1$ is an isometry and not a unitary.} corresponding to measuring $|f_1 \rangle \langle f_1|$ at time $t_1$. It is easily checked that $E^1_{f_1}$ is a positive operator on $\HH_S$ and that $\sum_{f_1} E^1_{f_1} = \id_S$. Therefore $\{ E^1_{f_1} \}$ is a positive operator-valued measure (POVM). Similarly, we have
\begin{equation} 
p(f_2) = \tr( V_2^\dagger (|f_2 \rangle \langle f_2|_F \otimes \id_{SW}) V_2 \rho) := \tr(E_{f_2}^2 \rho). \label{eq:p_f2_V}
\end{equation}
The calculation of the POVM elements yields
\begin{align}
E^1_{U} &= \ket \uparrow \bra \uparrow \label{E1U} \\
E^1_{D} &= \ket \downarrow \bra \downarrow \label{E1D} 
\end{align}
and
\begin{align}
E^2_{U} &= |a|^2 |\phi_1 \rangle \langle \phi_1| + |b|^2  |\phi_2 \rangle \langle \phi_2| \label{eq:E^2_U} \\
E^2_D &= |b|^2 |\phi_1 \rangle \langle \phi_1| + |a|^2  |\phi_2 \rangle \langle \phi_2|. \label{eq:E^2_D}
\end{align}

Assumptions P1 and P3 imply that there exists a joint POVM $\{G_{f_1 f_2} \}$ such that
\begin{equation}
p(f_1, f_2) = \tr ( G_{f_1 f_2} \rho),
\end{equation}
and requiring that the marginals obey P2 for all states means that $\sum_{f_1} G_{f_1 f_2} = E^1_{f_1}$ and $\sum_{f_2} G_{f_1 f_2} = E^2_{f_2}$. When there exists such a $\{ G_{f_1 f_2} \}$, the POVMs $\{E^1_{f_1} \}$ and $\{E^2_{f_2} \}$ are called jointly measurable.

If (at least) one of the two POVM's considered is sharp, then joint measureability is equivalent to commutativity, and there is a unique joint observable $G_{f_1 f_2} = E_{f_1}^1 E_{f_2}^2$ with the correct marginals (Proposition 8 of Ref.~\cite{Heinosaari2008}). Since we are considering two-outcome POVMs, and since $E^1$ given by Eqs.~\eqref{E1U} and \eqref{E1D}, is sharp, joint measurability is equivalent to $[E_U^1, E_U^2] = 0$. Direct calculation yields
\begin{align}
[E^1_U, E^2_U] &= (|a|^2 - |b|^2) a b^* \ket \uparrow \bra \downarrow \nonumber \\
&+ (|b|^2 - |a|^2) a^* b \ket \downarrow\bra \uparrow. \label{eq:commutator}
\end{align}
 So these two POVMs are not jointly measurable for general choices of $a,b$, which concludes the proof. \qed
\end{proof}

In the following we make some conceptual remarks about the proof of the above theorem. Even though our proof uses the language of joint measurability, due to a formal equivalence with that problem, the physical interpretation of joint measurability is different in our scenario. Indeed, joint measurability usually refers to the possibility to "simultaneously" measure two POVMs via a third joint POVM. Two non-jointly measureable POVMs can nevertheless be measured sequentially, one after the other, but the first measurement will generally disturb the state which is the input to the second measurement. In our Wigner's friend scenario, we are likewise considering measurement operators that correspond to observed outcomes at two subsequent times. However, assumption P2 implies Eq.~\eqref{eq:p_f2_V}, where $\rho$ is not affected by the measurement at $t_1$. Thus imposing P2 leads to a bypassing of the standard information-disturbance relations~\cite{Busch2009} and to a contradiction with assumptions P1 and P3.

Furthermore, note that it is not essential for Wigner to perform any measurement in order to derive a no-go result. Indeed, Wigner could instead perform a  "Hadamard" unitary $|\uparrow, U\rangle \mapsto \frac{1}{\sqrt{2}} (|\uparrow, U\rangle + |\downarrow, D\rangle)$,  $|\downarrow, D\rangle \mapsto \frac{1}{\sqrt{2}} (|\uparrow, U\rangle - |\downarrow, D\rangle)$, and the same theorem would follow after making the necessary modifications for the state $|\Psi(t_2)\rangle$.

\subsection{Motivation of the assumptions}
\label{sec:assumptions}
We attempt here to motivate each of the assumptions of our no-go theorem. We can only offer plausibility arguments, since, as we have already shown, these assumptions cannot in general all hold true in quantum mechanics. It should also be noted that the assumptions are not logically independent: for example one cannot hold P3 without at the same time assuming P1.

As discussed in Secs.~\ref{sec:gedanken} and \ref{sec:probs}, P1 is motivatived by the requirement that quantum theory --as any other predictive theory-- should provide us with conditional probability distributions for the friend's perceptions before and after Wigner's measurement, i.e. $p(f_2|f_1)$, and P2 provides single time probabilities for the friend's perception: $p(f_1)= \tr\left( \proj{f_1}_F \proj{\Psi(t_1)} \right)$ and $p(f_2)= \tr\left( \proj{f_2}_F \proj{\Psi(t_2)} \right)$. Thus we should in principle be able to construct a joint probability distribution for the friend's perceived outcomes at two different times, $p(f_1,f_2)$. Even if one initially only assigns probabilities to events directly perceived by an observer, such as $f_1$ and $f_2$, the requirement of predictability for a theory leads us to assign probabilities to the joint event $(f_1, f_2)$, although this is not a directly perceivable event in its own (arguably, one cannot have direct perceptions about two different times). 

Assumption P1 can also be understood as a special case of the general assumption that measurement records are \textit{facts of the world}~\cite{Brukner2017}, or of the \textit{absoluteness of observed events} (AOE) -- the assumption that "an observed event is a real single event, and not relative to anything or anyone"~\cite{Bong2020, Cavalcanti2020} -- applied to events $f_1$ and $f_2$. It is important to emphasize that the negation of AOE is not necessarily the claim that measurements outcomes are observer-dependent. Indeed the observed events in assumption P1 are all associated with the same observer, and thus P1 is conceptually different from the version of AOE used in deriving the no-go theorem of Ref.~\cite{Bong2020}, which is about joint probability assignments for the measurement outcomes of multiple observers.

P2 can be justified by appealing to the belief that interpretations of quantum mechanics should be empirically equivalent, i.e., that they all yield the same experimental predictions. Since P2 definitely holds in certain interpretations, notably in the Everett interpretation~\cite{Everett1957}, one should then expect P2 to hold in general. Since P2 can in principle be tested empirically, it is appropriate to regard quantum mechanics with objective collapse~\cite{Bassi2013} as a different physical theory from unitary quantum mechanics, and not merely a different interpretation~\cite{Baumann2018}. 

Assumption P3 can be understood as a conservative extension of the Born rule --which assigns single-time probabilities linearly in the quantum state-- to joint events at multiple times: P3 asks that the joint probabilities for events at multiple times must also depend linearly in the initial quantum state. P3 is true in typical laboratory situations where the usual quantum mechanical state-update rule can be used to calculate probabilities. Moreover, P3 can be motivated operationally, in a way that is customary in the context of generalized probabilistic theories~\cite{Hardy2001, Barrett2007}. We can imagine that a third agent is preparing the initial state of the system qubit, independently from the friend and Wigner. One might assume that, after fully specifying all relevant details for the friend's and Wigner's measurement setups (this includes the measurement basis for both of them, the initial quantum state of the friend, etc.), the probabilities $p(f_1, f_2)$ only depend on the quantum state $\rho_S$ but not on the way that the state was prepared. Suppose that $p_\sigma(f_1, f_2)$ and $p_\tau(f_1, f_2)$ are the probability distributions when the system state $\sigma$ or $\tau$ is prepared. Since $\rho = \lambda \sigma + (1 - \lambda) \tau $ can be prepared by tossing a biased coin which leads to prepare $\sigma$ with probability $\lambda$, and $\tau$ otherwise, the linearity of probability implies that $p_\rho (f_1, f_2) = \lambda p_\sigma(f_1, f_2) + (1 - \lambda) p_\tau (f_1, f_2)$.~\footnote{A further, independent justification for P3, i.e. allowing probabilistic mixtures, is that it implies that optimal compression is equivalent to linear compression~\cite{Hardy2009_foliable, Hardy2011}} Roughly speaking, upholding P1 while denying P3 amounts to the claim that quantum mechanics is "incomplete", in the sense that a full specification of the initial state $\rho_S$ is not sufficient for computing $p(f_1,f_2)$. Furthermore, a convincing case against P3 should involve the prescription and justification of a specific non-linear two-time probability rule; Bohmian mechanics is an example of this strategy, as we discuss further in Section~\ref{sec:interpretations}.

\section{Implications of the no-go theorem}

\subsection{The no-go theorem in different interpretations of quantum mechanics}
\label{sec:interpretations}

As mentioned above, strategies for coping with the no-go theorem Thm.~\ref{thm:no_joint}, i.e. deciding which of the assumptions one is most likely to drop, will depend on one's interpretation of quantum theory. We believe that organizing interpretations according to which of the assumptions they reject can help to give a clearer understanding of the fundamental differences between them. In what follows we will go through each of the assumptions and for each give examples of a prominent interpretation that would reject it. We do not strive here for exhaustivity, but rather to give an impression of the variety of ways in which our theorem can be understood. In the interest of space, our representation of any interpretation will be rather superficial.

\begin{enumerate}
\item [P1] According to our understanding, the Everett (or \textit{many-worlds}) interepretation~\cite{Everett1957, wallace2012emergent} 
denies that it is meaningful to assign a joint probability $p(f_1, f_2)$ to the friend's observations at multiple times.\footnote{In Chapter 7 of Ref. \cite{wallace2012emergent}, Wallace concludes that the there are only two viable candidates for a correct theory of identity (i.e. for what it means to talk about the "same object" at two different times) within the many-worlds interpretation; he calls these canditates the \textit{Lewisian view} and the \textit{Stage view}. In the Lewisian view, the identity of an object holds over a period of time within a (decohered) history, while in the Stage view the identity of objects only refers to a single instant in time. When "worlds" are allowed to interefere with each other as in the Wigner's friends thought experiment, the Lewisian view appears less viable.} This is at least Bell's diagnostic: 
\begin{quote}
"Everett [...] tries to associate each particular branch at the present time with some particular branch at any past time in a tree-like structure, in such a way that each representative of an observer has actually lived through the particular past that he remembers. In my opinion this attempt does not succeed and is in any case against the spirit of Everett's emphasis on memory contents as the important thing. We have no access to the past, but only to present memories."~\cite{Bell2004}
\end{quote}
Only in situations where a sufficient amount of decoherence is present is it possible to identify "worlds" branching in time, which would allow to meaningfully speak of $p(f_1, f_2)$. By construction, this is not the case in Wigner's friend scenario.

Further, note that operational approaches~\cite{Brukner2017} might only allow for the assignment of probabilities that can in principle be measured by performing many trials of the experiment, or in situations where probability assignments can be related to rational bets. This is not the case for the joint event $(f_1,f_2)$ here, because in a Wigner's friend experiment there is no "reliable record" of $f_1$ that remains available after time $t_2$. 

\item[P2] There are at least two ways that this assumption can be denied: objective collapse of the wave function, or subjective collapse of the wave function. In a theory with objective collapse~\cite{Bassi2013}, not only would P2 be false, but the predictions that Wigner makes using Eq.~\eqref{eq:psi_t2} would be verifiably wrong. In a theory with subjective state assignments such as QBism~\cite{Fuchs2010}, an agent is normatively constrained to use the Born rule for computing probabilities, but the quantum state used for doing so is up to the agent's good judgment; furthermore, QBism prohibits agents from assigning quantum states to themselves~\cite{DeBrota2020}. Therefore there can be subjective collapse: the friend would have the right to use the usual state update rule in order to calculate $p(f_2|f_1)$ -- and thus not recover $p(f_2)$ according to Eq.~\eqref{eq:p_fi} -- while Wigner uses unitary evolution for his own predictions.

\item[P3] In the de Broglie-Bohm interpretation~\cite{Bohm1952_1, Bohm1952_2, durr2009bohmian}, the memory of the friend has a definite and observer-independent value at all times and P1 holds. Furthermore, it can be proven that Bohmian mechanics recovers the same single-time predictions as unitary quantum mechanics so that P2 holds~\cite{Bohm1952_1, Bohm1952_2}. Therefore it must be P3 that fails to hold in that interpretation. Indeed, in is known in the context of double-slit intereference that the Bohmian guidance equation is non-linear in the density operator~\cite{Bell2004_Bohm, Luis2015}. It would be interesting for future work to calculate $p(f_1, f_2)$ for this experiment within a Bohmian description.
\end{enumerate}

\subsection{The (non-)persistence of memory in special cases}
\label{sec:memory}

Theorem~\ref{thm:no_joint} has been derived by assuming that in Wigner's choice of measurement basis, the parameters $a,b$ are arbitrary. However, from Eq.~\eqref{eq:commutator} it is easy to see which particular choices of $a,b$ allow the assumptions P1-P3 to be satisfied. In the following we compute  $p(f_1,f_2)$ for these special cases. We shall see (point 2 below) that assumptions P1-P3 can lead to counterintuitive conclusions about the time-evolution of the friend's memory even in those cases. There are essentially two possibilities which make the commutator in Eq.~\eqref{eq:commutator} vanish, which is equivalent to satisfying all three assumptions.
\begin{enumerate}
\item $|a|=1, b= 0$. (The case $a=0, |b|=1$ differs by a relabeling of the basis states). This corresponds to Wigner performing a measurement of the friend and system in the "computational basis" $|1\rangle = |\uparrow, U\rangle$, $|2\rangle = |\downarrow, D\rangle$\, revealing to him which result the friend observed. In that case the unique probability distribution that satisfies the assumptions of the theorem is  $p(f_1, f_2) = \tr ( E_{f_1}^2 E_{f_2}^2 \rho)$, with
\begin{align}
E_U^1 &= E_U^2 =  \ket \uparrow \bra \uparrow\\
E_D^1 &= E_D^2 =  \ket \downarrow \bra \downarrow.
\end{align}
It is easy to verify that $p(f_2|f_1) = \delta_{f_1 f_2}$, which means that the friends memory of the outcome is perfectly preserved.

\item $|a|^2 = |b|^2 = \frac{1}{2}$. This corresponds to Wigner performing a measurement in the ``Bell basis'' , for example $|1\rangle = \frac{1}{\sqrt{2}} (|\uparrow, U\rangle + |\downarrow, D \rangle)$, $|2\rangle = \frac{1}{\sqrt{2}} (|\uparrow, U\rangle - |\downarrow, D \rangle)$.\footnote{Here again we restrict our analysis to the only two out of the four ``Bell's states'' that are physically relevant in the described scenario.} Eqs.~\eqref{eq:E^2_U} and~\eqref{eq:E^2_D} show that the relative phases do not matter, so it suffices to consider this example. We have in this case $p(f_1, f_2) = \tr(E^1_{f_1} E^2_{f_2} \rho)$, with
\begin{align}
E_U^1 = \ket \uparrow \bra \uparrow\\
E_D^1 = \ket \downarrow \bra \downarrow\\
E^2_U = E^2_D = \frac{\id}{2},
\end{align}
and one can check that $p(f_2 | f_1) = \frac{1}{2}$. This means that the friend's memory gets flipped with probability $\frac{1}{2}$, independently of the initial state $\rho$. This is particularly surprising in the case where the initial state is $|\psi\rangle = \frac{1}{\sqrt{2}}( \ket \uparrow +\ket \downarrow)$, because in that case Wigner performs a \textit{non-disturbance measurement}~\cite{Baumann2019generalized,baumann2020wigner}. This means that the joint state of the friend and system $|\Psi(t_1) \rangle$ is actually an eigenstate of Wigner's measurement. One might expect that in this case, since the quantum state is not changed by Wigner's measurement, the friend's perceived result should remained unchanged as well; this is implicitly assumed in most discussions on the Wigner's friend thought experiment (and explicitly, for example, in Ref.~\cite{baumann2020wigner}). However, this conflicts with the assumption of quantum mechanical linearity of probabilities: if $p(f_1 ,f_2)$ is linear in $\rho$, the friend's perceived outcome must get flipped with probability $\frac{1}{2}$, indepently of $\rho$.

\end{enumerate}

\section{Conclusion}

From the point of view of Wigner, assuming that the friend's memory has a (unknown but definite) value is akin to assuming a hidden-variable model. Bell-type arguments involving two Wigners and two friends~\cite{Brukner2017, Brukner2018, Bong2020} have shown that if we further make a locality assumption on that hidden-variable model, it will not be possible to reproduce the quantum mechanical predictions. In this paper, we have shown that even from the friend's perspective, treating the memory of her measurement outcome as having a value throughout the experiment is in conflict with important features of quantum mechanics. More precisely, we have shown that it is not possible to assign a joint probability to her observed outcomes at two different times of the thought experiment, in a way that is compatible with unitary marginal probabilities and with the linear dependence of quantum mechanical probabilities on quantum states.

How to understand this theorem will depend on one's interpretation of quantum mechanics, but it seems that interesting lessons can be drawn from various interpretational points of view. Many popular interpretations (excluding hidden-variable interpretations like Bohmian mechanics) implicitly satisfy the principle that legitimate probability assignments should depend linearly on the initial quantum state. It appears in light of our theorem that the consequence of such a commitment is that one must in general either prohibit the use of present information to predict the future (drastically scaling down the predictive power of quantum theory), or deny that unitary quantum mechanics makes valid single-time predictions on all scales. That such a radical conclusion is necessary \textit{in general} does not affect the fact that for all practical purposes, i.e. in normal conditions when sufficient amounts of decoherence is present, one can continue to successfully use present information for predictions.

Our results might also raise interesting questions about the persistence of identity for the friend. If it is not possible for the friend to use the Born rule --or any other rule linear in the quantum state of the system-- to assign a joint probability distribution to her observed outcomes before and after Wigner's measurement, then to what extent can the friend at these two different times be considered the same agent? It is conceivable, although counterintuitive, that the friend at $t_1$ and the friend at $t_2$ should be legitimately considered to be two distinct agents.~\footnote{Conceptually speaking, this would be a costly conclusion to make in general, since these "two agents" share many common memories about their past.} In that case one could reach similar conclusions to the ones of Ref.~\cite{Cavalcanti2020}, and say that the friend's outcome at $t_2$ is \textit{not an event} from the point of view of the friend at $t_1$, and vice-versa.\\

\section{Acknowledgements}
We thank Mateus Ara\'{u}jo for comments on a previous version of the manuscript. V.B. acknowledges support from CoQuS, Vienna Doctoral School (VDS) and the QUOPROB project (no. I 2906- G24) of the Austrian Science Fund (FWF). F.D.S. acknowledges the financial support through a DOC Fellowship of the Austrian Academy of Sciences (OAW). We acknowledge support of the Austrian Science Fund (FWF) through the SFB project "BeyondC", a grant from the Foundational Questions Institute (FQXi) Fund and a grant from the John Templeton Foundation (Project No. 61466) as part of the The Quantum Information Structure of Spacetime (QISS) Project (qiss.fr). The opinions expressed in this publication are those of the authors and do not necessarily reflect the views of the John Templeton Foundation. Research at Perimeter Institute is supported in part by the Government of Canada through the Department of Innovation, Science and Economic Development Canada and by the Province of Ontario through the Ministry of Colleges and Universities.

\end{document}